\documentclass[10pt,letter,journal]{IEEEtran}
\usepackage{subcaption}
\usepackage{cite}
\usepackage{amsmath,amsthm,amssymb,amsfonts}
\usepackage{lipsum}
\usepackage{bbm}

\usepackage{mathtools}
\usepackage{bm}
\usepackage{algorithmic}
\usepackage{graphicx}
\usepackage{textcomp}
\usepackage{tikz}
\usepackage{xcolor}
\usepackage{svg}
\usepackage{flushend}
\usepackage{balance}
\usepackage[spaces,hyphens]{url}
\renewcommand{\baselinestretch}{0.97}
\theoremstyle{plain}
\newtheorem{lemma}{Lemma}
\newtheorem{corollary}{Corollary}
\newtheorem{theorem}{Theorem}
\theoremstyle{definition}

\newtheorem{assum}{Assumption}
\newtheorem{remark}{Remark}
\newtheorem{proposition}{Proposition}
\newcommand{\remove}[1]{}

\title{Optimal Threshold Type Policies for Partially Observable Restless Bandits}
\author{Anu Krishna, Rahul Meshram and Kesav Ram Kaza
\thanks{Anu Krishna (e-mail: vasanthakuma@alum.iisc.ac.in ) is alumnus of Department  of Electronics Systems  Engineering, IISc, Bangalore, R. Meshram is with the Department of Electrical Engineering, Indian Institute of Technology Madras, Chennai, India. (e-mail: {rahulmeshram}@ee.iitm.ac.in), and K. R. Kaza is with the Indian Institute of Information Technology, Allahabad, India (e-mail: krkaza@iiita.ac.in).  Work of R. Meshram is supported from IITM NFIG Grant, IITM NFSG Grant, and ANRF grant Project No EEQ/2021/000812.}}

\begin{document}


\maketitle

\begin{abstract}
We study a finite-state partially observable restless multi-armed bandit (PO-RMAB) motivated by resource-constrained wildlife monitoring. Underlying condition of each location evolves independently, while only a limited number of locations can be actively monitored at each decision epoch. Activation reveals the current state, whereas passive operation provides no observation, yielding a collapsing-bandit belief dynamics. Our main contribution is a structural characterization of optimal policies for the  multidimensional belief-state problem. We establish sufficient conditions under which an activation advantage  is monotone with respect to the belief state, and hence the optimal policy has the threshold type structure over the $(M-1)$-dimensional belief simplex. The key result is obtained by bounding the variation of the value function through a model-dependent Lipschitz constant. 
We further specialize the result to one-step birth--death dynamics and derive explicit bounds. 
\end{abstract}

\section{Introduction}




Restless multi-armed bandit (RMAB) is a class of sequential decision problem where $N$ arms are independent and each arm is modeled as finite state Markov process. The state of each arm evolves at each time step and the evolution of each arm is depends on action performed on an arm. Each action corresponds to play or not play of an arm. Each action on an arm yields reward which depends on state of that arm. Arms are coupled with budget constraint at the decision maker. The goal of decision maker is to determine the play of arms in each time step to maximize the long term expected discounted cumulative reward.  In classical RMAB \cite{Whittle88},  the popular approach is Whittle index policy, where budget constrained are relaxed and Lagrangian relaxation is introduced with introduction of subsidy for not playing an arm. The indexability is key criteria for the index policy, in this as subsidy increases there is monotonicity of passive set.   
A key step toward establishing indexability is to characterize the optimal policy for each subsidy as a threshold policy in the belief state. Indexability additionally requires the passive set to expand monotonically with the subsidy.  

In this paper, we consider a finite-state partially observable RMAB (PO-RMAB), where the state of an arm is unobserved when the arm is passive but is perfectly revealed when it is activated. The state evolves according to the action selected by the decision maker. The decision maker therefore maintains a belief distribution over the finite state space, which is updated according to the action and observation. For the two-state case, this observation structure is commonly referred to as the \textit{collapsing bandit} model \cite{Mate2020Bandit}.
PO-RMAB models arise in several resource-constrained monitoring and
decision-making applications. Most existing structural results focus on
two-state models, while multi-state models are considerably more
challenging because the belief state lies in a multidimensional simplex.
We address this challenge by deriving model-dependent Lipschitz
conditions for an increasing threshold structure.


\begin{figure}[t]
    \centering
    \includegraphics[width=0.85\linewidth]{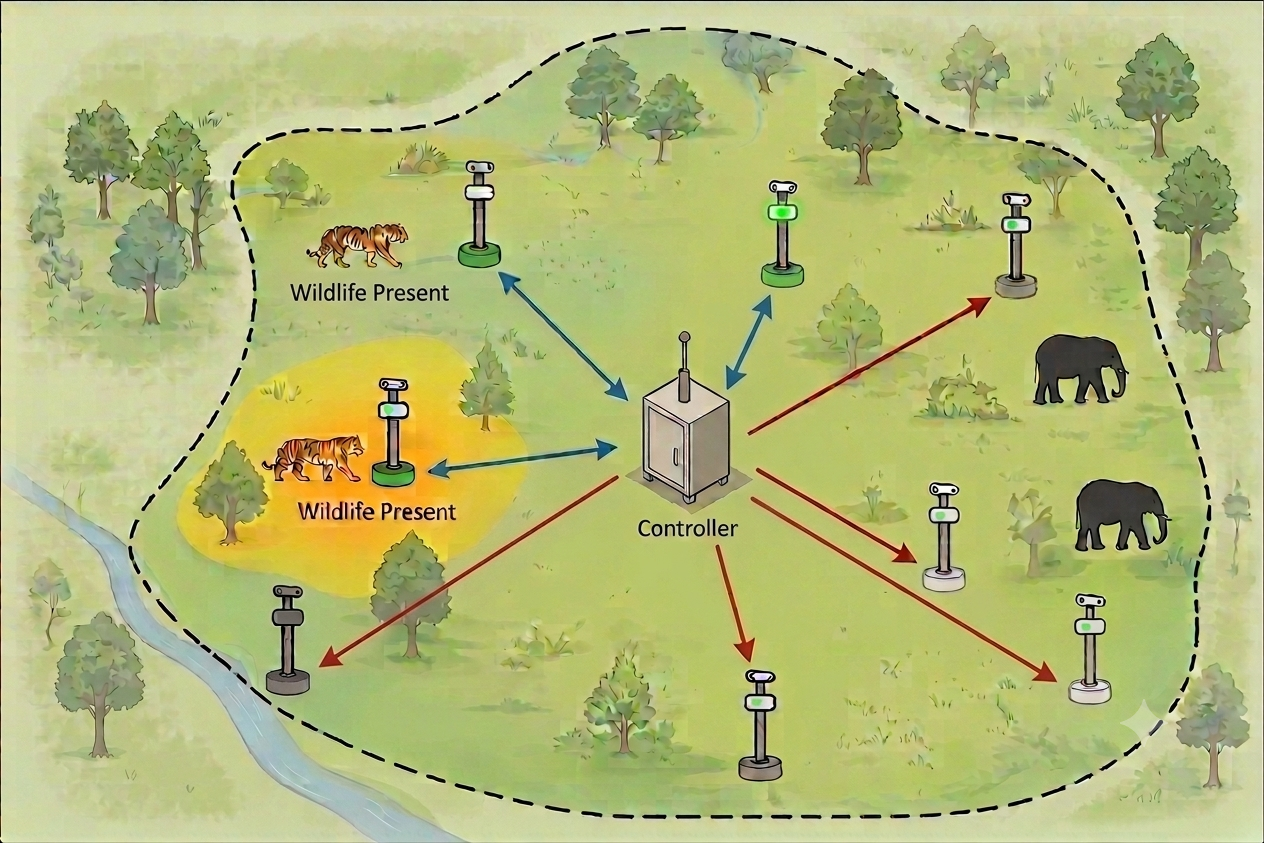}
    \caption{ A schematic diagram of a partially observable, resource-constrained wildlife monitoring system. 
    }
    \label{fig:SchematicDiagram}
\end{figure}


 In wildlife monitoring,  monitoring rely on sensor networks that is deployed across large, remote areas. These sensors are used to detect the presence of endangered species and to mitigate human-wildlife conflicts such as crop damage and train collisions \cite{qian2016restless}. 
  Since sensing devices are typically constrained by battery capacity, communication bandwidth, and data-processing resources, it is generally infeasible to activate all sensing locations simultaneously. A central controller must therefore repeatedly allocate a limited number of sensing or intervention actions among multiple locations. This creates a resource-constrained sequential decision problem in which the condition of each location evolves over time, while its current state may not be directly observable unless the location is activated.
 This setting naturally leads to a partially observable RMAB problem. Each location is associated with an underlying stochastic process describing its wildlife activity or threat condition. When a location is not activated, its state continues to evolve, but the controller does not obtain a direct observation of the current state. The controller must therefore maintain a belief distribution over the possible states and use this information to allocate the available resources. Importantly, activation can provide information about the current condition of a location while also affecting its subsequent state evolution. Thus, the controller must jointly account for uncertainty, information acquisition, resource allocation, and future system dynamics. We formulate this problem as PO-RMAB with finite states. The objective is to allocate these limited monitoring resources across
locations so as to maximize the expected discounted cumulative
monitoring reward. An schematic of the monitoring system and
the corresponding decision-making process is illustrated in Fig.~\ref{fig:SchematicDiagram}.

\vspace{-5pt}
\subsection{Related Work}
We present related work on RMAB, PO-RMAB, and its applications. 

\textbf{Restless multi-armed bandit (RMAB):} The RMAB framework was introduced by Whittle \cite{Whittle88}, who also
proposed the Whittle index policy based on a Lagrangian relaxation of the
resource constraint. The broader literature on Markovian multi-armed
bandits is reviewed in \cite{Gittins11}. The Whittle index policy is
particularly attractive because of its asymptotic optimality under
appropriate conditions \cite{Weber1990index}. A key requirement for
applying the Whittle index is indexability, which is often established by
showing that the optimal single-arm policy has a threshold structure in
the state and that the passive set expands monotonically with the subsidy.
Simulation-based approaches for computing and evaluating Whittle indices
are studied in \cite{Meshram2020}. More generally, RMABs can be viewed as
a class of weakly coupled Markov decision processes, for which Lagrangian
relaxation provides a useful decomposition and heuristic policies
\cite{adelman2008relaxations,hawkins2003lagrangian}.

\textbf{Partially observable restless multi-armed bandit (PO-RMAB):} PO-RMABs have been extensively studied for two state Markov model under various model assumptions. The belief about state is scalar and it is updated using Bayes' rule. Indexability and threshold policies have been studied in \cite{LiuZhao10,Mate2020Bandit,Meshram18}. 
PO-RMAB models arise in applications including wildlife protection \cite{qian2016restless}, intrusion detection \cite{liu2012dynamic}, opportunistic scheduling in multi-channel communications \cite{LiuZhao10,Meshram15}, public healthcare interventions \cite{Mate2020Bandit,nino2026adherence}, machine repair \cite{Akbarzadeh2019Restless}, and wildfire monitoring and planning \cite{kaza2024constrained}. Most existing results for these applications consider two-state models. Multi-state PO-RMABs are considerably more challenging because the belief state lies in a multidimensional simplex, making the characterization of threshold policies nontrivial. Existing multi-state results have focused on specific models such as restart bandits \cite{Akbarzadeh2022indexability,Meshram2021indexability}, relaxed indexability \cite{Liu2025relaxed}, and multi-action formulations \cite{Meshram2025:MultiAction}.  However, threshold structures over the multidimensional belief space are not characterized. In this work, we study threshold policies for finite-state PO-RMABs with state-revealing activation. Since the belief state lies in a multidimensional simplex, we use Lipschitz continuity of the value function to derive model-dependent sufficient conditions for an increasing threshold structure. We further specialize these conditions to one-step birth--death dynamics.

\subsection{Our Contributions}

The main contributions of this paper are as follows.
\begin{enumerate}
    \item \textbf{Finite-state threshold structure:}
For finite-state partially observable collapsing restless bandits, we establish sufficient conditions for an increasing threshold structure over the $(M-1)$-dimensional belief simplex under monotone activation-reward differences, increasing rewards, and appropriate stochastic ordering of transition matrices.

\item \textbf{Lipschitz-based sufficient conditions:}
Using Lipschitz continuity of the value function, we derive a model-dependent sufficient condition that relates the minimum adjacent activation-reward difference to the reward span, the Dobrushin contraction coefficient, and the diameter of the post-activation belief set.


\item \textbf{Explicit conditions for one-step birth-death dynamics:}
For one-step birth-death transition matrices, we explicitly characterize the Dobrushin coefficient and post-activation belief-set diameter for $M=2$, $M=3$, and $M\geq4$, yielding explicit threshold conditions in terms of local transition probabilities and reward span. probabilities and reward span.
\end{enumerate}

The paper is organized as follows. Section~\ref{sec:model} describes the mathematical model of PO-RMAB and problem formulation. In Section~\ref{sec:Decoupled-RMAB}, we discuss decoupled PO-RMAB. We analyze a single-armed restless bandit in Section~\ref{sec:SARB}, we study Lipschitz property of value and threshold type policy in belief state. Finally, in Section~\ref{sec:conclusion}, we provide concluding remarks. 
\vspace{-5pt}
\section{Mathematical Model}
\label{sec:model}

We consider a resource-constrained wildlife monitoring system with $N$ independent 
monitored locations. Each location is represented as an arm $n$, $1 \leq n \leq N$, and is modeled as a
partially observable Markov decision process (POMDP) $\mathcal{M}_n =
\left\{
\mathcal{S}_n,\mathcal{A}_n,\mathcal{P}_n,
\mathcal{R}_n,\mathcal{O}_n,\mathcal{Z}_n,\beta
\right\}.$

In this problem, each arm has $M$ possible states and two actions, with state $0$ representing the lowest level of wildlife activity (or threat) and state $M-1$ representing the highest level. The state and action
spaces are given by $\mathcal{S}_n = \{0,1,\ldots,M-1\},
\text{and }
\mathcal{A}_n = \{0,1\}$, respectively, where action $1$ denotes active (scheduled) operation and
$0$ denotes passive (unscheduled) operation. For arm $n$, let $P_n^a = \left[p_{n,ij}^a\right]_{i,j=0}^{M-1},
\quad a\in\{0,1\},$
denote the transition probability matrix under action $a$, where $p_{n,ij}^a
=
\Pr\left(s_n(t+1)=j \mid s_n(t)=i,a_n(t)=a\right).$
Thus, for every $i$ and $a$, $\sum_{j=0}^{M-1}p_{n,ij}^a=1.$

At time $t$, arm $n$ occupies state
$s_n(t)\in\mathcal{S}_n$, the controller selects an action $a_n(t)\in\mathcal{A}_n$ and the resulting reward is
$r_n(s_n(t),a_n(t))$. The discount factor is denoted by $\beta$, and 
$0<\beta<1$.

The state of an arm is observed when the arm is active, whereas no
observation is available when the arm is passive. That is, when
$a_n(t)=1$, the observation space is
$\mathcal{O}_n^1=\{0,1,\ldots,M-1\}.$ 
Under the assumption of perfect observation upon activation,
$o_n(t)=s_n(t)$. When $a_n(t)=0$, the observation space is
$\mathcal{O}_n^0=\{\Phi\}.$ The observation kernel $Z_n(o|i,1)=\mathbbm{1}\{o=i\} $ and $Z_n(\Phi|i,0)=1$. Then the joint observation space is given by $\mathcal{O}(a)=\prod_{n=1}^{N}\mathcal{O}_n^{a_n}$. The history accumulated through time $t$ is $H(t)=\{\mathbf{a}(1),\mathbf{o}(1),\dots,\mathbf{a}(t-1),\mathbf{o}(t-1)\},$
using $\mathbf{a}(t)=\{a_1(t),\dots,a_N(t)\}$ actions of arms and $\mathbf{o}(t)=\{o_1(t),\dots,o_N(t)\}$ observation of arms.

Since the state is not directly observable under passive operation, the
planner maintains a belief vector over the state space. The belief of arm
$n$ at time $t$ is
\[
\boldsymbol{\omega}_n(t)
=
\left[
\omega_n^0(t),\omega_n^1(t),\ldots,
\omega_n^{M-1}(t)
\right]^{\mathsf T},
\]
where $\omega_n^i(t)
=
\Pr\left(
s_n(t)=i\mid H(t),\boldsymbol{\omega}_n(0)
\right).$
The belief belongs to the probability simplex
\begin{equation*}
\Pi
=
\left\{
\boldsymbol{\omega}\in\mathbb{R}^M:
\omega^i\geq 0,\;
\sum_{i=0}^{M-1}\omega^i=1
\right\}.
\end{equation*}
Hence, the belief space is an $(M-1)$-dimensional simplex.

\subsection{Belief Update}
For arm $n$, the belief at time $t$ is defined as $\omega_n^s(t)=\Pr\left(s_n(t)=s\mid H(t),\boldsymbol{\omega}_n(0)\right)$, where $s\in\mathcal{S}_n$. The belief vector is given by
$\boldsymbol{\omega}_n(t).$

Suppose that the active action $a_n(t)=1$ is selected and that the current state is observed to be $s_n(t)=k$, where $k\in\{0,1,\ldots,M-1\}$. Since $p_{n,kj}^{1}=\Pr\left(s_n(t+1)=j\mid s_n(t)=k,a_n(t)=1\right)$, the belief of the next state, conditioned on observing state $k$, is given by $\boldsymbol{\omega}_n(t+1)=\boldsymbol{\Gamma}_{n,k},$
where 
\[
\boldsymbol{\Gamma}_{n,k}
=
\begin{bmatrix}
p_{n,k0}^{1} \\
p_{n,k1}^{1} \\
\vdots \\
p_{n,k,M-1}^{1}
\end{bmatrix}
=(\mathbf{P}_n^1)^T \mathbf{e}_k,\quad k \in \{0,1,\ldots,M-1\},
\]
and $\mathbf{e}_k$ denotes the $k$-th standard basis vector.
\remove{
For the three-state model, where transitions are allowed only between neighboring states, we have $p_{n,02}^{1}=p_{n,20}^{1}=0$. Therefore,
\[
\boldsymbol{\Gamma}_{n,0}
=
\begin{bmatrix}
p_{n,00}^{1} \\
p_{n,01}^{1} \\
0
\end{bmatrix}.
\]
Similarly,
\[
\boldsymbol{\Gamma}_{n,1}
=
\begin{bmatrix}
p_{n,10}^{1} \\
p_{n,11}^{1} \\
p_{n,12}^{1}
\end{bmatrix}
=
\begin{bmatrix}
p_{n,10}^{1} \\
p_{n,11}^{1} \\
1-p_{n,10}^{1}-p_{n,11}^{1}
\end{bmatrix}.
\]
Finally,
\[
\boldsymbol{\Gamma}_{n,2}
=
\begin{bmatrix}
0 \\
p_{n,21}^{1} \\
p_{n,22}^{1}
\end{bmatrix}.
\]
}
When the passive action $a_n(t)=0$ is selected, no state observation is obtained. The belief therefore evolves according to the passive transition matrix $\boldsymbol{\omega}_n(t+1)=(\mathbf{P}_n^0)^T \boldsymbol{\omega}_n(t).$
Equivalently, for each $j\in\{0,1,\ldots,M-1\}$, $\omega_n^j(t+1)=\sum_{i=0}^{M-1} \omega_n^i(t) p_{n,ij}^{0}.$
\subsection{Problem Formulation}

A policy $\phi$ maps the history of observations and actions to the
joint action. That is,
$\phi: H(t)\to \mathcal{A}.$
Let $\mathbf{a}(t)
=
(a_1(t),\ldots,a_N(t))$
denote the joint action at time $t$. The admissible action set is
\begin{equation}
\mathcal{A}
=
\left\{
\mathbf{a}\in\{0,1\}^N:
\sum_{n=1}^{N}a_n\leq B
\right\},
\label{eq:joint_action_set}
\end{equation}
where $B$ is the maximum number of arms that can be activated in each
time slot.

For a joint belief state $\boldsymbol{\Omega}
=
\left(
\boldsymbol{\omega}_1,\ldots,\boldsymbol{\omega}_N
\right),$
the expected instantaneous reward from arm $n$ under action $a_n$ is $R_n(\boldsymbol{\omega}_n,a_n)
=
\sum_{i=0}^{M-1}
\omega_n^i r_n(i,a_n).$

The objective is to maximize the expected discounted total reward
\begin{equation*}
V(\boldsymbol{\Omega})
=
\max_{\phi}
\mathbb{E}_{\phi}
\left[
\sum_{t=0}^{\infty}
\beta^t
\sum_{n=1}^{N}
R_n
\left(
\boldsymbol{\omega}_n(t),a_n(t)
\right)|\boldsymbol{\Omega}(0)=\boldsymbol{\Omega}
\right].
\label{eq:joint_value_function}
\end{equation*}

The corresponding Bellman equation is
\begin{equation*}
\begin{aligned}
V(\boldsymbol{\Omega})
=
\max_{\mathbf{a}\in\mathcal{A}}
\Bigg\{
&
\sum_{n=1}^{N}
R_n(\boldsymbol{\omega}_n,a_n)
\\
&+
\beta
\sum_{\mathbf{o}\in\mathcal{O}(a)}
\Pr(\mathbf{o}\mid\boldsymbol{\Omega},\mathbf{a})
V\left(
\tau(\boldsymbol{\Omega},\mathbf{o},\mathbf{a})
\right)
\Bigg\}.
\end{aligned}
\label{eq:joint_bellman}
\end{equation*} Since the arms evolve independently and observations are generated independently conditional on the corresponding arm states and actions, the joint observation distribution factorizes as $\Pr(\mathbf{o}\mid\boldsymbol{\Omega},\mathbf{a})
=
\prod_{n=1}^{N}
\Pr
\left(
o_n\mid\boldsymbol{\omega}_n,a_n
\right).$ The resulting belief-state POMDP has a high-dimensional state space,
which makes direct computation of the optimal value function
computationally intractable for large $N$.

\section{Decoupling of the Restless Multi-Armed Bandit}
\label{sec:Decoupled-RMAB}

We relax the instantaneous activation constraint in Eqn.~\eqref{eq:joint_action_set} to an expected discounted activation constraint:
\begin{equation}
\begin{aligned}
\max_{\phi}\quad
&
\mathbb{E}_{\phi}
\left[
\sum_{t=0}^{\infty}
\beta^t
\sum_{n=1}^{N}
R_n(\boldsymbol{\omega}_n(t),a_n(t))
\right]
\\
\text{s.t.}\quad
&
\mathbb{E}_{\phi}
\left[
\sum_{t=0}^{\infty}
\beta^t
\sum_{n=1}^{N}a_n(t)
\right]
\leq
\frac{B}{1-\beta}.
\end{aligned}
\label{eq:relaxed_problem}
\end{equation}

Introducing a Lagrange multiplier $\lambda\geq0$ for the relaxed
constraint gives
\begin{equation}
\begin{aligned}
V^\lambda(\boldsymbol{\Omega})
=
\max_{\phi}
\mathbb{E}_{\phi}
\Bigg[
\sum_{t=0}^{\infty}\beta^t
\Bigg(
&
\sum_{n=1}^{N}
R_n(\boldsymbol{\omega}_n(t),a_n(t))
\\
&
+
\lambda
\left(
B-\sum_{n=1}^{N}a_n(t)
\right)
\Bigg)
\Bigg].
\end{aligned}
\label{eq:lagrangian_problem}
\end{equation}

We write the Bellman equation as
\begin{equation}
\begin{aligned}
V^\lambda(\boldsymbol{\Omega})
&=
\max_{\mathbf{a}\in\{0,1\}^N}
\Bigg\{
\sum_{n=1}^{N}
R_n(\boldsymbol{\omega}_n,a_n)
+
\lambda
\left(
B-\sum_{n=1}^{N}a_n
\right)
\\
&
+
\beta
\sum_{\mathbf{o}\in\mathcal{O}}
V^\lambda
\left(
\tau(\boldsymbol{\Omega},\mathbf{o},\mathbf{a})
\right)
\prod_{n=1}^{N}
\Pr
\left(
o_n\mid\boldsymbol{\omega}_n,a_n
\right)
\Bigg\}.
\end{aligned}
\label{eq:lagrangian_bellman}
\end{equation}

\begin{lemma}
\label{lemma:value_decomposition}
The optimal value function of the Lagrangian-relaxed problem can be
decomposed as
$V^\lambda(\boldsymbol{\Omega})
=
\sum_{n=1}^{N}
V_n^\lambda(\boldsymbol{\omega}_n)
+
\frac{B\lambda}{1-\beta},$
where $V_n^\lambda(\boldsymbol{\omega}_n)$ is the optimal value
function of the single-armed problem
\begin{equation}
\begin{aligned}
V_n^\lambda(\boldsymbol{\omega}_n)
&= 
\max_{a_n\in\{0,1\}}
\Bigg\{
R_n(\boldsymbol{\omega}_n,a_n)
+
\lambda(1-a_n)
\\
& +
\beta
\sum_{o_n\in\mathcal{O}_n}
V_n^\lambda
\left(
\tau_n(\boldsymbol{\omega}_n,o_n,a_n)
\right)
\times
\Pr
\left(
o_n\mid\boldsymbol{\omega}_n,a_n
\right)
\Bigg\}.
\end{aligned}
\label{eq:single_arm_bellman}
\end{equation}
\end{lemma}

The proof follows along the lines of
\cite[Proposition~1]{adelman2008relaxations}, which is for weakly-coupled MDPs. Detailed proof for the PO-RMAB case is given~\cite{Meshram2025:MultiAction}[Appendix B]. 
Thus, the Lagrangian
relaxation decouples the original $N$-armed problem into $N$
independent single-armed restless bandit problems, each of which has a binary action space (activation and non-activation). The Lagrange multiplier $\lambda$ is now called subsidy (for non-activation of the arm).

\section{Restless Single-Armed Bandit}
\label{sec:SARB}
We now study the partially observable restless single armed bandit (PO-RSAB). For notational simplicity, we omit the arm
index $n$. The belief state belongs to the simplex
\begin{equation}
\Pi
=
\left\{
\boldsymbol{\omega}\in\mathbb{R}^M:
\omega^i\geq0,\;
\sum_{i=0}^{M-1}\omega^i=1
\right\}.
\end{equation}

For a given subsidy $\lambda$, the action-value function corresponding
to the passive action $a=0$ is
\begin{equation}
Q^\lambda(\boldsymbol{\omega},0)
=
R(\boldsymbol{\omega},0)
+
\lambda
+
\beta
V^\lambda
\left(
\tau(\boldsymbol{\omega})
\right),
\label{eq:Q_passive}
\end{equation}
where $\tau(\boldsymbol{\omega})
=
(P^0)^{\mathsf T}\boldsymbol{\omega}.$

Under the active action $a=1$, the state is observed. Therefore,
\begin{equation}
Q^\lambda(\boldsymbol{\omega},1)
=
R(\boldsymbol{\omega},1)
+
\beta
\sum_{i=0}^{M-1}
\omega^i
V^\lambda(\boldsymbol{\Gamma}_i),
\label{eq:Q_active}
\end{equation}
where $\boldsymbol{\Gamma}_i
=
(P^1)^{\mathsf T}\mathbf{e}_i.$ Therefore, from Eqn.~ \eqref{eq:Q_passive},\eqref{eq:Q_active} and \eqref{eq:single_arm_bellman},
\begin{equation}
V^\lambda(\boldsymbol{\omega})
=
\max_{a\in\{0,1\}}
Q^\lambda(\boldsymbol{\omega},a).
\label{eq:single_arm_value}
\end{equation}

The Lagrange multiplier $\lambda$ is interpreted as a subsidy for
leaving the arm passive.

\begin{lemma}
\label{lemma:convexity}
For a fixed $\lambda$, the optimal value function
$V^\lambda(\boldsymbol{\omega})$ is convex in
$\boldsymbol{\omega}\in\Pi$.
\end{lemma}

The proof follows by induction and uses
\cite[Lemma~2]{Astrom69}.

\textbf{Monotone Likelihood Ratio Order:}
For two beliefs
$\boldsymbol{\omega},\boldsymbol{\omega}'\in\Pi$, we say that
$\boldsymbol{\omega}$ dominates $\boldsymbol{\omega}'$ in the
monotone likelihood ratio (MLR) order, denoted by
$\boldsymbol{\omega}
\succeq_{\mathrm{MLR}}
\boldsymbol{\omega}',$
if $\omega^j\omega^{\prime i}
\geq
\omega^i\omega^{\prime j},
\quad
0\leq i<j\leq M-1.$

\textbf{TP2 Transition Matrices:} A nonnegative matrix
$P=[p_{ij}]_{i,j=0}^{M-1}$ is said to be totally positive of order two
(TP2) if every $2\times2$ minor is non-negative. That is, $p_{ij}p_{k\ell}
-
p_{i\ell}p_{kj}
\geq0,
\quad
0\leq i<k\leq M-1,
0\leq j<\ell\leq M-1.$

A standard property of TP2 matrices is that  $\boldsymbol{\omega}
\succeq_{\mathrm{MLR}}
\boldsymbol{\omega}'
\quad\Longrightarrow\quad
P^{\mathsf T}\boldsymbol{\omega}
\succeq_{\mathrm{MLR}}
P^{\mathsf T}\boldsymbol{\omega}'.$ Consequently, if $P^0$ is TP2, then $\boldsymbol{\omega}
\succeq_{\mathrm{MLR}}
\boldsymbol{\omega}'
\quad\Longrightarrow\quad
\tau(\boldsymbol{\omega})
\succeq_{\mathrm{MLR}}
\tau(\boldsymbol{\omega}').$ Under the corresponding TP2 assumptions on the active transition and
observation model, the post-observation beliefs satisfy $\boldsymbol{\Gamma}_0
\preceq_{\mathrm{MLR}}
\boldsymbol{\Gamma}_1
\preceq_{\mathrm{MLR}}
\cdots
\preceq_{\mathrm{MLR}}
\boldsymbol{\Gamma}_{M-1}.$

\textbf{Assumptions for Monotonicity:} 
Define $\delta_i:=r(i,1)-r(i,0).$
We impose the following assumptions.

\begin{assum}[Bounded rewards]
\label{assumption:monotonicity}
For every state-action pair,
\begin{equation*}
|r(i,a)|\leq r_{\max}<\infty,
\quad
i\in\{0,1,..,M-1\},\quad a\in\{0,1\}.
\end{equation*}
\end{assum}
\begin{assum}[Monotone rewards]\label{assumption:reward_nondecreasing}
  For each action $a\in\{0,1\}$, $r(i,a)\leq r(j,a)$, $\forall 0\leq i<j \leq M-1.$

\end{assum}

\begin{assum}[Monotone action advantage]
  $\delta_i\leq \delta_j$, $\forall 0\leq i<j \leq M-1.$

\label{assumption:d_monotonicity}
\end{assum}
\color{black}
\begin{assum}[TP2 transitions]
    Both $P^0$ and $P^1$ are TP2.
\end{assum}
\begin{lemma}[Monotonicity of the Value Function]
\label{lemma:monotonicity}
Under Assumption~$1,$ $2,$ $3$ and $4$, the optimal value
function is monotone with respect to $\bm{\omega}$ in the MLR order. In particular, $\boldsymbol{\omega}
\succeq_{\mathrm{MLR}}
\boldsymbol{\omega}'
\quad\Longrightarrow\quad
V^\lambda(\boldsymbol{\omega})
\geq
V^\lambda(\boldsymbol{\omega}').$

\end{lemma}

The proof follows along the lines of
\cite[Proposition~1]{Lovejoy87}.

\begin{remark}
Since $P^1$ is TP2 and the value function is monotone, we have $V^\lambda(\boldsymbol{\Gamma}_i)
\leq
V^\lambda(\boldsymbol{\Gamma}_j),
\quad 0 \leq i \leq j \leq M-1.$
\end{remark}
Define $\mathcal{C}_1(\boldsymbol{\omega})
=
\sum_{i=0}^{M-1}
\omega^i
V^\lambda(\boldsymbol{\Gamma}_i).$ Since $P^1$ is TP2, its rows are ordered in the MLR sense, $\boldsymbol{\Gamma}_0\preceq_{\mathrm{MLR}}\boldsymbol{\Gamma}_1\preceq_{\mathrm{MLR}}
\cdots\preceq_{\mathrm{MLR}}\boldsymbol{\Gamma}_{M-1}.$ Since $V^\lambda$ is increasing with respect to the MLR order,$V^\lambda(\boldsymbol{\Gamma}_0)\leq V^\lambda(\boldsymbol{\Gamma}_1)\leq\cdots\leq V^\lambda(\boldsymbol{\Gamma}_{M-1}).$ Moreover, since MLR dominance implies first-order stochastic dominance, $\boldsymbol{\omega}\succeq_{\mathrm{MLR}}\boldsymbol{\omega}'\Longrightarrow
\mathcal{C}_1(\boldsymbol{\omega})\geq \mathcal{C}_1(\boldsymbol{\omega}').$
Similarly, define $\mathcal{C}_0(\boldsymbol{\omega})
=
V^\lambda
\left(
\tau(\boldsymbol{\omega})
\right).$ Since $P^0$ is TP2, the belief transition $\tau$ preserves the MLR order. Thus, by the MLR monotonicity of $V^\lambda$, $\boldsymbol{\omega}\succeq_{\mathrm{MLR}}\boldsymbol{\omega}'\Longrightarrow
\mathcal{C}_0(\boldsymbol{\omega})\geq\mathcal{C}_0(\boldsymbol{\omega}').$

\subsection{Action-Value Difference and Threshold Policy}

Define the difference between the active and passive action values as $\Delta_{\lambda}(\boldsymbol{\omega})=Q^{\lambda}(\omega,1)-Q^{\lambda}(\omega,0)$. From  \eqref{eq:Q_passive} and \eqref{eq:Q_active}, This is called the activation advantage function.

\begin{equation}
\Delta_{\lambda}(\boldsymbol{\omega})
=\sum_{i=0}^{M-1}
\omega^i \delta_i-\lambda+
\beta
\left[
\sum_{i=0}^{M-1}
\omega^i
V^\lambda(\boldsymbol{\Gamma}_i)-V^\lambda(\tau(\boldsymbol{\omega}))
\right].
\label{eq:delta_expanded}
\end{equation}

We seek to show that $\Delta_{\lambda}(\boldsymbol{\omega})$ is non-decreasing with respect to $\bm{\omega}$ in the MLR order for fixed $\lambda$. This monotonicity implies that activation becomes weakly more attractive as the belief increases in the MLR order, and hence the optimal policy has an increasing threshold structure on the belief space.

Although MLR monotonicity of $V^\lambda$ establishes monotonicity of
the individual continuation terms, it does not, in general, imply
monotonicity of their difference in~\eqref{eq:delta_expanded}. Thus,
additional sufficient conditions are required to establish a
threshold structure for the general $M$-state model.

\subsection{Lipschitz-Based Sufficient Conditions}
\label{sec:m-state-lipschitz-condition}

We now extend the Lipschitz-based sufficient conditions for a threshold-type
policy. For the post-activation belief vectors
$\bm{\Gamma}_i
=
P_{i,\cdot}^{1},
\quad i \in \{0,\ldots,M-1\},$
define their $\ell_1$-diameter by $D_{\Gamma}
:=
\max_{i,j\in\{0,\ldots,M-1\}}
\left|
\bm{\Gamma}_i-\bm{\Gamma}_j
\right|_1.$
For each action $a\in\{0,1\}$, define the reward span by
\begin{equation}
\label{eq:m-state-reward-span}
\operatorname{span}(r^a)
:=
\max_{i\in\{0,\ldots,M-1\}}r(i,a)
-
\min_{i\in\{0,\ldots,M-1\}}r(i,a).
\end{equation}

The corresponding Lipschitz constant of the one-step reward is $L_R
:=
\frac{1}{2}
\max_{a\in\{0,1\}}
\operatorname{span}(r^a).$

Define the Dobrushin contraction coefficient of each transition matrix by
\begin{equation}
\label{eq:m-state-ca-general}
c_a
:=
\frac{1}{2}
\max_{i,j\in\{0,\ldots,M-1\}}
\left|
P_{i,\cdot}^{a}-P_{j,\cdot}^{a}
\right|_1,
\qquad a\in\{0,1\},
\end{equation}
and let $c
:=
\max\{c_0,c_1\}.$ Since $\frac{1}{2}
\max_{i,j}
\left|
\bm{\Gamma}_i-\bm{\Gamma}_j
\right|_1
\leq 1,$
we have $0\leq c\leq1.$

\begin{lemma}[Lipschitz Property of Value function]
\label{lem:m-state-lipschitz}
Suppose that the transition kernels satisfy the corresponding
Dobrushin contraction condition. Then the value function satisfies $\left|
V^\lambda(\bm{\omega})
-
V^\lambda(\bm{\omega}')
\right|
\leq
L_V
\left|
\bm{\omega}-\bm{\omega}'
\right|_1,$
where $L_V
=
\frac{L_R}{1-\beta c}.$
\end {lemma}
\begin{proof}
    The proof follows from the contraction property of the Bellman operator. For details, see Appendix \ref{appendix: Lemma 4}.
\end{proof}

Recall that $\delta_i
=
r(i,1)-r(i,0),
\quad
i\in \{0,\ldots,M-1\},$
and its belief expectation $\delta(\bm{\omega})
:=
\sum_{i=0}^{M-1}\delta_i\omega^i.$Recall the assumption \ref{assumption:d_monotonicity}. 
We now state main result. 

\begin{theorem}[Monotonicity of Activation Advantage Function]
\label{lem:m-state-threshold-sufficiency}
For two ordered beliefs $\bm{\omega}\succeq_{MLR}\bm{\omega}',$ under assumptions $1-4$
if $\delta(\bm{\omega})-\delta(\bm{\omega}')
\geq
\frac{\beta L_Vc_0}{2}
\left|
\bm{\omega}-\bm{\omega}'
\right|_1,$
then $\Delta_{\lambda}(\bm{\omega})
\geq
\Delta_{\lambda}(\bm{\omega}').$  
\end{theorem}
\begin{proof}
    See Appendix \ref{appendix: lemma 5}.
\end{proof}
\begin{remark}
Monotonicity of activation advantage function in belief state implies that there exists the  threshold type optimal policy.  
\end{remark}

Define the minimum adjacent reward-difference margin, $m_\delta
:=
\min_{i=0,\ldots,M-2}
\left(\delta_{i+1}-\delta_i\right).$

\begin{proposition}[Explicit $M$-state sufficient condition]
\label{prop:m-state-explicit-condition}
Under the assumptions of Lemma~\ref{lem:m-state-threshold-sufficiency},
suppose $m_\delta
\geq
\frac{
\beta L_RD_\Gamma
}{
1-\beta c
}.$Then the active advantage $\Delta_{\lambda}(\bm{\omega})
=
Q^{\lambda}(\bm{\omega},1)-Q^{\lambda}(\bm{\omega},0)$
is increasing with respect to the belief order. 
\end{proposition}

\begin{proof}
    See Appendix \ref{appendix: proposition 1}.
\end{proof}

\begin{corollary}
\label{cor:m-state-fully-explicit}
Suppose activation reveals the current state, so that $\bm{\Gamma}_i=P_{i,\cdot}^1.$
Then $D_\Gamma
=\max_{i,j\in\{0,\ldots,M-1\}}
\sum_{k=0}^{M-1}
\left|
P_{ik}^1-P_{jk}^1
\right|$ and $c=\frac12
\max_{\substack{a\in\{0,1\}\\
i,j\in\{0,\ldots,M-1\}}}
\sum_{k=0}^{M-1}
\left|
P_{ik}^a-P_{jk}^a
\right|.$
Therefore, a sufficient condition for an increasing activation advantage is
\begin{equation}
\label{eq:m-state-fully-explicit}
m_\delta
\geq
\frac{
\displaystyle
\frac{\beta}{2}
\max_{i,j}
\sum_{k=0}^{M-1}
\left|
P_{ik}^1-P_{jk}^1
\right|
}{
\displaystyle
1-
\frac{\beta}{2}
\max_{\substack{a\in\{0,1\}\\i,j}}
\sum_{k=0}^{M-1}
\left|
P_{ik}^a-P_{jk}^a
\right|
}
\,
\max_{a\in\{0,1\}}
\operatorname{span}(r^a).
\end{equation}
\end{corollary}

\subsubsection{Special Cases and Interpretations}
\label{sec:special-cases}

We now apply the general sufficient condition from 
Proposition~\ref{prop:m-state-explicit-condition} (or 
Corollary~\ref{cor:m-state-fully-explicit}) to several important 
special cases.

\paragraph{Common Transition Matrix}

Suppose both actions have the same transition matrix $P^0 = P^1 = P$. Let $c_P = \frac{1}{2} \max_{i,j\in\{0,\ldots,M-1\}} \sum_{k=0}^{M-1} |P_{ik} - P_{jk}|.$
Then $c_0 = c_1 = c_P$, and the condition in 
Corollary~\ref{cor:m-state-fully-explicit} reduces to $m_\delta \geq \frac{\beta c_P}{1 - \beta c_P}\max_{a \in \{0,1\}} \operatorname{span}(r^a),$ where $m_\delta$ is the minimum adjacent increase in the active advantage.

Furthermore, if $\delta_i = \delta_0 + m \cdot i, \quad \forall i \in \{0,\ldots,M-1\},$then $m_\delta = m$, and the condition becomes
$m \geq \frac{\beta c_P}{1 - \beta c_P}\max_{a \in\{0,1\}} \operatorname{span}(r^a).$

\paragraph{Identical Rows}

Suppose $P_{i,\cdot}^a = P_{j,\cdot}^a$ for all $0 \leq i,j \leq M-1$ and 
$a \in \{0,1\}$. Then $c_0 = c_1 = 0$, and the sufficient condition 
reduces to $m_\delta \geq 0$. Thus, in this special case, increasing passive advantage alone is sufficient
for an increasing threshold policy.
\subsection{Threshold Conditions for One-Step Random-Walk Transitions}
\label{sec:m-state-random-walk-threshold}

We now apply the general sufficient condition from 
Proposition~\ref{prop:m-state-explicit-condition} to a specific family 
of birth-death transitions.

Consider an $M$-state birth-death model. Let the state space be
$\mathcal{S} = \{0,1,\ldots,M-1\}$, and for each action $a \in \{0,1\}$, 
let the one-step transition matrix be $P^a = [P_{i,j}^a]_{i,j=0}^{M-1},$

where
\[
P_{i,j}^a =
\begin{cases}
p_a, & j = i+1,\\
q_a, & j = i-1,\\
1 - p_a - q_a, & 1 \leq i = j \leq M-2,\\
1 - p_a, & i = j = 0,\\
1 - q_a, & i = j = M-1,\\
0, & \text{otherwise},
\end{cases}
\]
and $p_a, q_a \geq 0$, $p_a + q_a \leq 1$.

Further, assume that the rewards are increasing in the state. That is, $r(i,a)\leq r(j,a), \quad \forall a \in\{0,1\}, 0\leq i<j \leq M-1$.

Define the reward range $S_R := \max_{a \in \{0,1\}} (r({M-1},a) - r(0,a)).$ Since the rewards are increasing, the $\ell_1$-Lipschitz constant of the 
expected one-step reward is
\[
L_R = \max_{a \in \{0,1\}} \frac{r(M-1,a) - r(0,a)}{2} = \frac{S_R}{2}.
\]

\subsubsection{Dobrushin Coefficient for $M \geq 4$}

For $M \geq 4$, the first and last rows of $P^a$ are $P^a_{0,\cdot} = (1-p_a, p_a, 0, \ldots, 0)$
 and $P^a_{M-1,\cdot} = (0, \ldots, 0, q_a, 1-q_a).$
Since their supports are disjoint, $\lVert P^a_{0,\cdot} - P^a_{M-1,\cdot}\rVert_1 = 2.$ Since the $\ell_1$ distance between two probability distributions is at 
most 2, it follows that, $M \geq 4,$ $c_a = \frac{1}{2} \max_{i,j} \lVert P^a_{i,\cdot} - P^a_{j,\cdot}\rVert_1 = 1.$
Therefore, for $M \geq 4,$ $c = \max_{a \in \{0,1\}} c_a = 1.$

Suppose further that activation reveals the current state, so that 
$\bm{\Gamma}_i = P^1_{i,\cdot}$. Recall $D_\Gamma = \max_{i,j} \lVert \bm{\Gamma}_i - \bm{\Gamma}_j\rVert_1 = \max_{i,j} \lVert P^1_{i,\cdot} - P^1_{j,\cdot}\rVert_1.$

Since the first and last rows have disjoint supports, $M \geq 4$, $D_\Gamma =2.$

\begin{theorem}
\label{thm:m-state-random-walk-threshold}

Consider the $M$-state birth-death model described in \ref{sec:m-state-random-walk-threshold}
with $M \geq 4$. Suppose that:
\begin{enumerate}
    \item $p_a, q_a \geq 0$ and $p_a + q_a \leq 1$ for $a \in \{0,1\}$;
    \item the rewards are increasing in the state;
    \item the $\delta_i$ is increasing;
    \item activation reveals the current state and $\bm{\Gamma}_i = P^1_{i,\cdot}$.
\end{enumerate}

If $m_\delta \geq \frac{\beta}{1-\beta} S_R$
where $m_\delta = \min_i (\delta_{i+1} - \delta_i)$, then the activation 
advantage $\Delta_{\lambda}(\boldsymbol{\omega})$ 
is nondecreasing with respect to the
MLR order on the belief space. Therefore, the optimal policy admits an 
increasing threshold structure.
\end{theorem}
\begin{proof}
    See Appendix \ref{appendix: Theorem 1}.
\end{proof}







Hence, under the global Dobrushin coefficient approach, the explicit 
sufficient threshold condition takes the form:
\[
m_\delta \geq
\begin{cases}
\dfrac{\beta(1-p_1-q_1)}{1-\beta[1-\min\{p_0+q_0, p_1+q_1\}]} S_R, & M=2,\\[1.5em]
\dfrac{\beta(1-\min\{p_1, q_1\})}{1-\beta(1-m_P)} S_R, & M=3,\\[1.5em]
\dfrac{\beta}{1-\beta} S_R, & M \geq 4.
\end{cases}
\]
Here, $m_P := \min\{p_0, q_0, p_1, q_1\}.$
The $M \geq 4$ condition is more conservative than the corresponding 
$M=2$ and $M=3$ conditions because the global Dobrushin coefficient 
equals one once the extreme states have disjoint one-step transition supports.
\color{black}

\section{Numerical Illustrations }
\subsection{Optimal Policy: Two Dimensional Slice for $M=5$}
We numerically solve the discounted Bellman equation on the
five-state belief simplex $\Delta^4
=
\left\{
\boldsymbol{\omega}\in\mathbb{R}_+^5:
\sum_{i=0}^4 \omega_i=1
\right\}.$
The transition matrices under the passive and active actions are,
respectively,
\[
P^0
=
\begin{pmatrix}
0.40&0.30&0.15&0.10&0.05\\
0.20&0.35&0.25&0.15&0.05\\
0.10&0.20&0.40&0.20&0.10\\
0.05&0.15&0.25&0.35&0.20\\
0.05&0.10&0.15&0.30&0.40
\end{pmatrix},
\]
and
\[
P^1
=
\begin{pmatrix}
0.35&0.35&0.15&0.10&0.05\\
0.10&0.30&0.35&0.20&0.05\\
0.05&0.15&0.35&0.30&0.15\\
0.02&0.08&0.20&0.40&0.30\\
0.01&0.04&0.15&0.30&0.50
\end{pmatrix}.
\]
The corresponding state-dependent rewards are
$r^1=(0,1,2,3,4)$
and $r^0=(0,0.8,1.6,2.4,3.2).$
We use a discount factor $\beta=0.95$.

The belief simplex is discretized using a uniform grid with
resolution $25$, i.e., each belief component is an integer multiple
of $1/25$. For a given subsidy $\lambda$, value iteration is performed
according to
\[
V_{n+1}^{\lambda}(\boldsymbol{\omega})
=
\max\left\{
Q_n^{\lambda}(\boldsymbol{\omega},1),
Q_n^{\lambda}(\boldsymbol{\omega},0)
\right\},
\]
where $Q_n^{\lambda}(\boldsymbol{\omega},0)
=
\lambda+
\boldsymbol{\omega}r^{0}
+
\beta V_n(\boldsymbol{\omega}P^{0})$ and $Q_n^{\lambda}(\boldsymbol{\omega},1),
=
\boldsymbol{\omega}r^{1}
+
\sum_{i=0}^{4}
\omega^i
V_n^\lambda(\boldsymbol{\Gamma}_i)$, $\boldsymbol{\Gamma}_i
=
(P^1)^{\mathsf T}\mathbf{e}_i.$

Since the next belief at $a=0$ or
$1$ need not coincide with a grid point, we use
piecewise-linear interpolation over a Delaunay triangulation of the
discretized belief simplex. Since $\sum_{i=0}^{4}\omega_i=1$, only four belief coordinates are independent, and the triangulation is
constructed in the coordinates $(\omega_0, \omega_1, \omega_2, \omega_3)$.
The Delaunay triangulation and the corresponding barycentric
coordinates of all possible next-belief points are computed once and
reused throughout value iteration. Value iteration is terminated when $\|V_{n+1}-V_n\|_\infty < 10^{-6}.$ The subsidy is varied over $\lambda\in\{0.50,0.55,0.60,\ldots,2.05\}.$

For visualization, we consider the two-dimensional simplex obtained
by fixing $\omega_0=\omega_1=0.$ So $\sum_{i=2}^{4}\omega_i=1$, $\omega_i\geq 0,\quad i \in \{2,3,4\}$ which is a triangular simplex embedded in the
$(\omega_2,\omega_3,\omega_4)$ space. The three vertices correspond to $(1,0,0)$, $(0,1,0)$, and $(0,0,1)$. The triangle is colored according to the optimal action determined
from the sign of $\Delta_{\lambda}(\boldsymbol{\omega})$.
\begin{figure}[t]
    \centering
    \includegraphics[width=0.85\linewidth]{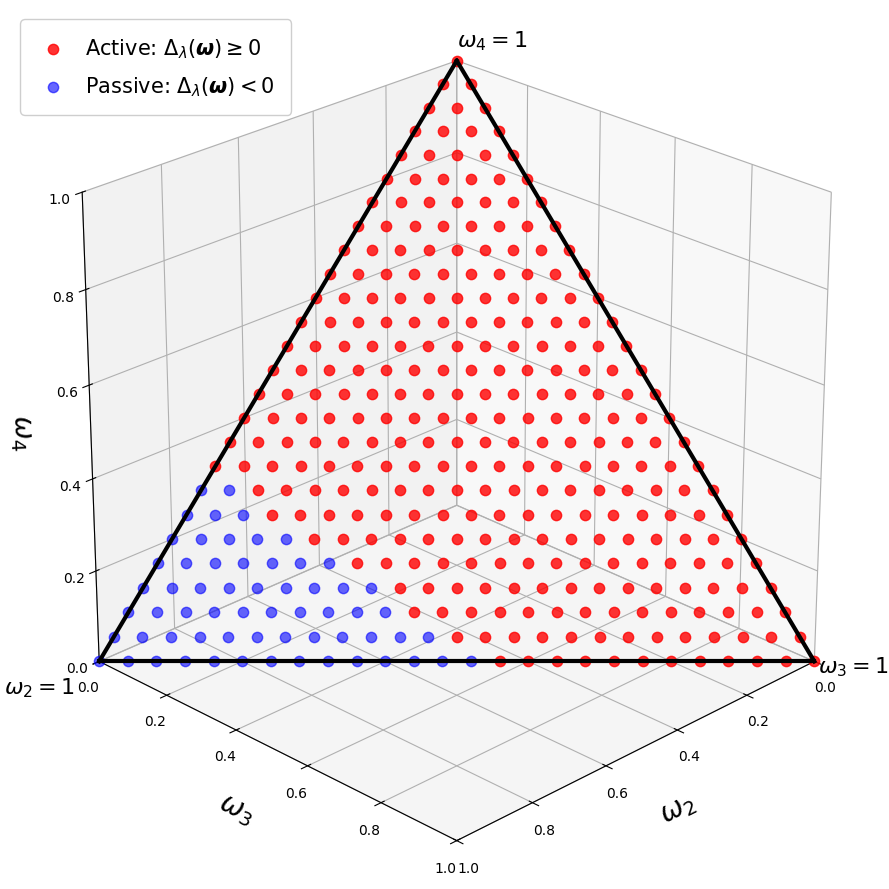}
    \caption{Active and passive regions on the two-dimensional belief
    simplex obtained by fixing $\omega_0=\omega_1=0$, so that
    $\omega_2+\omega_3+\omega_4=1$, for subsidy $\lambda=1.05$.
    Red points correspond to beliefs for which
    $\Delta_{\lambda}(\boldsymbol{\omega})\geq 0$ and the active action is optimal,
    whereas blue points correspond to
    $\Delta_{\lambda}(\boldsymbol{\omega})< 0$ and the passive action is optimal.}
    \label{fig:active_passive_regions}
\end{figure}

\section{Conclusion}
\label{sec:conclusion}
We studied a finite-state partially observable collapsing bandit with
state-revealing activation and established sufficient conditions for the
activation advantage to be increasing in the MLR order. This yields an
increasing threshold structure over the $(M-1)$-dimensional belief
simplex. We further specialized the conditions to one-step birth-death
models and obtained explicit bounds for different state-space sizes.

The threshold characterization provides a principled way to determine when a location should be activated based on its current belief, and can facilitate the design of scalable monitoring policies across multiple locations. The resulting threshold characterization also provides a structural foundation for Whittle-index-based policies. However, threshold structure for a fixed subsidy does not by itself establish indexability; subsidy monotonicity of the passive set remains an important direction for future work. A detailed application study, including application-specific modeling, parameter selection, and numerical evaluation for wildlife monitoring, is also left for future work.
\bibliographystyle{IEEEtran}
\bibliography{reference}

\appendix
\balance
\section*{Proof of Lemma 4}
\label{appendix: Lemma 4}
For any belief $\bm{\omega}$, the expected one-step reward under action $a$ is $r^a(\bm{\omega})
=
\sum_{i=0}^{M-1}\omega^i r(i,a).$
Since both $\bm{\omega}$ and $\bm{\omega}'$ are probability
distributions, $\sum_{i=0}^{M-1}(\omega^i-\omega^{\prime i})=0.$
Therefore, for any constant $c$, $r^a(\bm{\omega})-r^a(\bm{\omega}')
=
\sum_{i=0}^{M-1}
(\omega^i-\omega^{\prime i})(r(i,a)-c).$
Choosing $c=\frac{r_{\max}^a+r_{\min}^a}{2},$
where
$r_{\max}^a=\max_i r(i,a)$ and
$r_{\min}^a=\min_i r(i,a)$, gives $|r(i,a)-c|
\leq
\frac{1}{2}\operatorname{span}(r^a).$
Hence $\left|r^a(\bm{\omega})-r^a(\bm{\omega}')\right|
\leq
\frac{1}{2}\operatorname{span}(r^a)
\|\bm{\omega}-\bm{\omega}'\|_1
\leq
L_R\|\bm{\omega}-\bm{\omega}'\|_1.$

The transition operator contracts the $\ell_1$ distance by at most the
Dobrushin coefficient $c$. We initialize $V_0^\lambda(\bm{\omega})=0,$
and define the finite-horizon value iteration
\begin{equation*}
\label{eq:value-iteration-three-state}
V_{n+1}^\lambda(\bm{\omega})
=
\max_{a\in\{0,1\}}
Q_{n+1}^\lambda(\bm{\omega},a),
\end{equation*}
where
\begin{equation*}
\label{eq:q-value-three-state}
Q_{n+1}^\lambda(\bm{\omega},a)
=
R^\lambda(\bm{\omega},a)
+
\beta
\mathcal C_a
V_n^\lambda(\bm{\omega}).
\end{equation*}
Here, $R^\lambda(\bm{\omega},0)
=
R(\bm{\omega},0)+\lambda,$ and $R^\lambda(\bm{\omega},1)
=
R(\bm{\omega},1)$
and \(\mathcal C_a\) denotes the corresponding continuation operator.

Suppose that for $n$ th iteration  we have 
\begin{equation*}
\label{eq:induction-lipschitz}
|V_n^\lambda(\bm{\omega})
-
V_n^\lambda(\bm{\omega}')|
\le
L_n
\|
\bm{\omega}
-
\bm{\omega}'
\|_1.
\end{equation*}
Then, for every action \(a\),
\begin{align}
|Q_{n+1}^\lambda(\bm{\omega},a)
-
Q_{n+1}^\lambda(\bm{\omega}',a)|
& \le
|R^\lambda(\bm{\omega},a)
-
R^\lambda(\bm{\omega}',a)|\nonumber\\
&+
\beta
\left|
\mathcal C_aV_n^\lambda(\bm{\omega})
-
\mathcal C_aV_n^\lambda(\bm{\omega}')
\right|.
\label{eq:q-lipschitz-step}
\end{align}

The subsidy \(\lambda\) does not affect the Lipschitz modulus because it is
constant in the belief. We have $|R^\lambda(\bm{\omega},a)
-
R^\lambda(\bm{\omega}',a)|
\le
L_R
\|
\bm{\omega}
-
\bm{\omega}'
\|_1.$

Using the contraction property of the continuation operator, we get $\left|
\mathcal C_aV_n^\lambda(\bm{\omega})
-
\mathcal C_aV_n^\lambda(\bm{\omega}')
\right|
\le
c_aL_n
\|
\bm{\omega}
-
\bm{\omega}'
\|_1.$
Substituting in \ref{eq:q-lipschitz-step}, we obtain $|Q_{n+1}^\lambda(\bm{\omega},a)
-
Q_{n+1}^\lambda(\bm{\omega}',a)|
\le
\left(
L_R+\beta c_aL_n
\right)
\|
\bm{\omega}
-
\bm{\omega}'
\|_1.$

The \((n+1)\)-stage value function is $V_{n+1}^\lambda(\bm{\omega})
=
\max_{a\in\mathcal A}
Q_{n+1}^\lambda(\bm{\omega},a).$
We have 
\begin{align*}
\left|
V_{n+1}^\lambda(\bm{\omega})
-
V_{n+1}^\lambda(\bm{\omega}')
\right|
&=
\left|
\max_{a\in\mathcal A}
Q_{n+1}^\lambda(\bm{\omega},a)
-
\max_{a\in\mathcal A}
Q_{n+1}^\lambda(\bm{\omega}',a)
\right|\\
& \leq
\max_{a\in\mathcal A}
\left|
Q_{n+1}^\lambda(\bm{\omega},a)
-
Q_{n+1}^\lambda(\bm{\omega}',a)
\right|.
\label{eq:value-difference-max}
\end{align*}

Using the actionwise bound, we get 
\begin{align*}
\left|
V_{n+1}^\lambda(\bm{\omega})
-
V_{n+1}^\lambda(\bm{\omega}')
\right|
\leq
\max_{a\in\mathcal A}
\left\{
\left(
L_R+\beta c_aL_n
\right)
\left\|
\bm{\omega}-\bm{\omega}'
\right\|_1
\right\}.
\end{align*}
The term  $\left\|
\bm{\omega}-\bm{\omega}'
\right\|_1$
does not depend on the action, and hence 
\begin{align*}
\left|
V_{n+1}^\lambda(\bm{\omega})
-
V_{n+1}^\lambda(\bm{\omega}')
\right|
\leq
\left[
\max_{a\in\mathcal A}
\left(
L_R+\beta c_aL_n
\right)
\right]
\left\|
\bm{\omega}-\bm{\omega}'
\right\|_1.
\end{align*}

Define $c = \max_{a\in\mathcal A}c_a.$

Since \(L_R\), \(\beta\), and \(L_n\) do not depend on the action,
\begin{align*}
\max_{a\in\mathcal A}
\left(
L_R+\beta c_aL_n
\right)
&=
L_R
+
\beta L_n
\max_{a\in\mathcal A}c_a
\nonumber\\
&=
L_R+\beta cL_n.
\end{align*}

Consequently, 
\begin{align*}
    V_{n+1}^\lambda(\bm{\omega})
-V_{n+1}^\lambda(\bm{\omega}')
\leq
\left(
L_R+\beta cL_n
\right)
\left\|
\bm{\omega}-\bm{\omega}'
\right\|_1.
\end{align*}

It follows that a valid Lipschitz modulus for
\(V_{n+1}^\lambda\) is $L_R+\beta cL_n.$ If $L_{n+1}$ denotes the smallest Lipschitz modulus of
\(V_{n+1}^\lambda\), then $L_{n+1}
\leq
L_R+\beta cL_n.$

Since \(L_0=0\), repeated substitution gives
\begin{equation*}
\label{eq:Ln-geometric}
L_n
\le
L_R
\sum_{k=0}^{n-1}
(\beta c)^k.
\end{equation*}

Therefore,
\begin{equation*}
\label{eq:finite-horizon-L}
L_n
\le
L_R
\frac{1-(\beta c)^n}
{1-\beta c}.
\end{equation*} 
Since $0<\beta<1$ and $0\le c\le1,$
we have $\beta c<1.$ Letting \(n\to\infty\), we obtain $\frac{L_R}
{1-\beta c}.$
The discounted Bellman operator is a contraction, so
\(V_n^\lambda\) converges uniformly to \(V^\lambda\). Passing to the limit
gives
\[
L_V
\le
\frac{L_R}
{1-\beta c}.
\]
This proves Lemma \ref{lem:m-state-lipschitz}.
\section*{Proof of Theorem 1}
\label{appendix: lemma 5}

     Recall that
\begin{align*}
\Delta_{\lambda}(\bm{\omega})
-
\Delta_{\lambda}(\bm{\omega}') &=
\delta(\bm{\omega})
-
\delta(\bm{\omega}')-
\beta
\left[
V^\lambda(\tau(\bm{\omega}))
-
V^\lambda(\tau(\bm{\omega}'))
\right]\\
&+
\beta
\sum_{i=0}^{2}
(\omega_i-\omega_i')
V^\lambda(\Gamma_i).
\end{align*}

Suppose $\bm{\omega}
\succeq
\bm{\omega}',$ we want to show that $\Delta_{\lambda}(\bm{\omega})
-
\Delta_{\lambda}(\bm{\omega}') \geq 0.$ 

Define
\begin{align*}
D_{\lambda}^{0}(\bm{\omega},\bm{\omega}')
&:=
V^{\lambda}\bigl(\tau(\bm{\omega})\bigr)
-
V^{\lambda}\bigl(\tau(\bm{\omega}')\bigr),
\label{eq:D0-definition}
\\
D_{\lambda}^{1}(\bm{\omega},\bm{\omega}')
&:=
\sum_{i=0}^{2}
(\omega^i-\omega^{i,'})
V^{\lambda}(\Gamma_i).
\end{align*}
Then
\begin{eqnarray*}
\Delta_{\lambda}(\bm{\omega})
-
\Delta_{\lambda}(\bm{\omega}')
=
\delta(\bm{\omega})-\delta(\bm{\omega}')
+ 
\beta
\left[
D_{\lambda}^{1}(\bm{\omega},\bm{\omega}')
-
D_{\lambda}^{0}(\bm{\omega},\bm{\omega}')
\right].
\end{eqnarray*} 

Since $\mathbf{P}^1$ is TP2, the post-activation beliefs satisfy $\boldsymbol{\Gamma}_0
\preceq_{\mathrm{MLR}}
\boldsymbol{\Gamma}_1
\preceq_{\mathrm{MLR}}
\cdots
\preceq_{\mathrm{MLR}}
\boldsymbol{\Gamma}_{M-1}.$ Since $V^\lambda$ is nondecreasing with respect to the MLR order,
it follows that $V^\lambda(\boldsymbol{\Gamma}_0)
\leq
V^\lambda(\boldsymbol{\Gamma}_1)
\leq
\cdots
\leq
V^\lambda(\boldsymbol{\Gamma}_{M-1}).$ Hence, $D_{\lambda}^{1}(\boldsymbol{\omega},\boldsymbol{\omega}')
=
\sum_{i=0}^{M-1}
(\omega^i-\omega^{\prime i})
V^\lambda(\boldsymbol{\Gamma}_i)
\geq V^\lambda(\boldsymbol{\Gamma}_{M-1})\sum_{i=0}^{M-1}
(\omega^i-\omega^{\prime i})
=0.$

Since $D_{\lambda}^{1}(\bm{\omega},\bm{\omega}')\geq0$,
 $D_{\lambda}^{1}(\bm{\omega},\bm{\omega}')-D_{\lambda}^{0}(\bm{\omega},\bm{\omega}') \geq -D_{\lambda}^{0}(\bm{\omega},\bm{\omega}')$
. 

Using Lipschitz property of value function, we obtain the following result. 
\begin{align*}
\Delta_{\lambda}(\bm{\omega})
-
\Delta_{\lambda}(\bm{\omega}')
\ge
\delta(\bm{\omega})
-
\delta(\bm{\omega}')
-
\frac{\beta L_Vc_0}{2}
\|
\bm{\omega}
-
\bm{\omega}'
\|_1.
\end{align*} 
the right-hand side is nonnegative. Therefore,
\[
\Delta_{\lambda}(\bm{\omega})
\ge
\Delta_{\lambda}(\bm{\omega}'),
\]
which proves the monotonicity of the passive-action advantage.

\section*{Proof of Proposition 1}
\label{appendix: proposition 1}

For ordered beliefs, $\bm{\omega}\succeq\bm{\omega}',$
the increasing-difference property implies that the immediate reward
difference is increasing. In particular, its increase can be bounded below
in terms of the minimum adjacent increment $m_\delta$. Let $d_i=\omega^i-\omega'^i.$ Since $\bm\omega\succeq_{\rm MLR}\bm\omega'$, $\sum_{i=0}^k d_i\le 0,\qquad k=0,\ldots,M-2.$
Define $\alpha_k = \sum_{i=0}^k(\omega'^i-\omega^i) \ge0.$ Then, $\sum_{i=0}^{M-1}d_i\delta_i = \sum_{k=0}^{M-2} \alpha_k(\delta_{k+1}-\delta_k).$ Since, $\delta_{k+1}-\delta_k\ge m_\delta,$ $\delta(\bm\omega)-\delta(\bm\omega') \ge m_\delta\sum_{k=0}^{M-2}\alpha_k.$ Now, $\max_k\alpha_k = \frac12\|\bm\omega-\bm\omega'\|_1.$ Since, $\alpha_k\ge0$, $\sum_{k=0}^{M-2}\alpha_k
\ge
\max_k\alpha_k
=
\frac12\|\bm\omega-\bm\omega'\|_1.$

Therefore, $\delta(\bm{\omega})-\delta(\bm{\omega}')
\geq
\frac{m_\delta}{2}
\left\|
\bm{\omega}-\bm{\omega}'
\right\|_1.$

Using $L_V=\frac{L_R}{1-\beta c}$
and $m_\delta
\geq
\frac{
\beta L_RD_\Gamma
}{
1-\beta c
},$
the pairwise sufficient condition becomes
\[
\delta(\bm{\omega})-\delta(\bm{\omega}')
\geq
\frac{
\beta L_RD_\Gamma
}{
2(1-\beta c)
}
\left\|
\bm{\omega}-\bm{\omega}'
\right\|_1.
\]

 The conclusion follows from
Lemma~\ref{lem:m-state-threshold-sufficiency}.

\section*{Proof of Theorem 2}
\label{appendix: Theorem 1}
By Proposition~\ref{prop:m-state-explicit-condition}, it suffices to verify $m_\delta \geq \frac{\beta L_R D_\Gamma}{1-\beta c}.$

For $M\geq4$, the first and last rows of each transition matrix
$P^a$ have disjoint supports. Hence
$c_a
=
\frac{1}{2}\max_{i,j}
\left\|P^a_{i,\cdot}-P^a_{j,\cdot}\right\|_1
=1,$
and therefore $c=1$. Moreover, since
$\bm{\Gamma}_i=P^1_{i,\cdot}$, the same disjoint-support argument gives $D_\Gamma
=
\max_{i,j}
\left\|\bm{\Gamma}_i-\bm{\Gamma}_j\right\|_1
=2.$
Finally, by the definition of $L_R$, $L_R=\frac{S_R}{2}.$

Hence $\frac{\beta L_R D_\Gamma}{1-\beta c}
=
\frac{\beta(S_R/2)(2)}{1-\beta}
=
\frac{\beta}{1-\beta}S_R.$

Thus the condition
\[
m_\delta\geq\frac{\beta}{1-\beta}S_R
\]
is exactly the sufficient condition in
Proposition~\ref{prop:m-state-explicit-condition}. Therefore,
$\Delta_\lambda(\bm{\omega})$ is nondecreasing with respect to the
MLR order, and the optimal policy admits an increasing threshold
structure.

\end{document}